\documentclass[10pt,journal,twocolumn,twoside]{IEEEtran} 

\usepackage{graphicx}
\usepackage{epstopdf}
\usepackage{float}
\usepackage{array}
\usepackage{amsmath}
\usepackage{amssymb}
\usepackage{mdwmath}
\usepackage{mdwtab}
\usepackage{eqparbox}
\usepackage{stfloats}
\usepackage{fixltx2e}
\usepackage{cases} 

\usepackage{flushend}

\usepackage{xcolor}
\usepackage{makecell}
\usepackage{algpseudocode}
\usepackage[boxed,ruled,commentsnumbered]{algorithm2e}
\usepackage{url}
\usepackage{cite}
\ifCLASSOPTIONcompsoc
\usepackage[caption=false,font=normalsize,labelfont=sf,textfont=sf]{subfig}
\else
\usepackage[caption=false,font=footnotesize]{subfig}
\fi
\allowdisplaybreaks[4]

\newtheorem{corollary}{\bf Corollary}

\newtheorem{proof}{Proof}

\makeatletter
\renewcommand*{\@opargbegintheorem}[3]{\trivlist
      \item[\hskip \labelsep{\bfseries #1\ #2}] \textbf{(#3):}\ }
\makeatother

\begin{document}

\title
{Beamspace Modulation for Near Field Capacity Improvement in XL-MIMO Communications}
\author{ Shuaishuai Guo,~\IEEEmembership{Senior Member, IEEE}, and Kaiqian Qu,~\IEEEmembership{Student Member, IEEE}
\thanks{
Shuaishuai Guo and Kaiqian Qu are  with School of Control Science and Engineering, Shandong University, Jinan 250061, China and also with Shandong Provincial Key Laboratory of Wireless Communication Technologies (e-mail: shuaishuai\textunderscore guo@sdu.edu.cn, qukaiqian@mail.sdu.edu.cn).}
   }
\maketitle


\begin{abstract}
The spatial degrees of freedom (DoFs) greatly increase in the near-field region of millimeter wave or terahertz multiple-input multiple-output communications with extremely large antenna arrays (XL-MIMO). To employ the increased spatial DoFs,  a beamspace modulation (BM) strategy is introduced to the near field of XL-MIMO. BM can work  with a fixed small number of RF chains. It  exploits the increased spatial DoFs as modulation resources for capacity improvements. The achievable spectral efficiency and its asymptotic capacity are analyzed. Both theoretical and simulation results show that the proposed BM strategy considerably outperforms the existing benchmark that only selects the best beamspace for data transmission in terms of spectral efficiency.
\end{abstract}

\begin{IEEEkeywords}
Beamspace modulation, XL-MIMO, near field  communications
\end{IEEEkeywords}

\section{Introduction} 
\IEEEPARstart{T}{he} sixth generation (6G) mobile communications networks are envisioned to provide a more than ten-fold increase in spectral efficiency compared to the fifth generation (5G) communication networks. To achieve this goal, it is expected more and more antennas are equipped at the transceivers. Compared with the massive multiple-input multiple-output (MIMO) widely used in 5G, the number of antennas is becoming extremely large, known as extremely large MIMO (XL-MIMO), to support communications at high-frequency band, e.g., the millimeter wave (mmWave) frequency band or the terahertz frequency band \cite{dang2020should}. 

A sharp increase in antennas leads to a huge change not only in hardware but also in electromagnetic (EM) radiation. Specifically, the EM radiation is generally divided into near field and far field regions. The near field region grows large as the carrier frequency and array aperture  become large. The distance between the region boundary to the transmitter is called Fraunhofer distance\cite{Selvan2017FraunhoferAF}.
The distance is proportional to the product of the carrier frequency 
and the square of the array aperture. 
With mmWave/terahertz frequency band and XL-MIMO being applied, most communications may operate in the near field region. This is different from existing communications that operate in far field. In the far field, the EM radiation is modeled as planner waves. While in the near field, the EM radiation has to be modeled more accurately as spherical waves. 

Recently, research on near field communications has started to become a focus and seems to promise to be a key point for further improvements in communication performance \cite{Cui2022NearFieldCF,Cui2021ChannelEF,Han2019ChannelEF,Cui2021NearFieldCE,Wei2022ChannelEF,Lu2022NearFieldMA,Wu2022MultipleAF,Zhang2022BeamFF}. In \cite{Cui2022NearFieldCF},  the authors  identified challenges and possible future research directions for near field communications. Accurate channel estimation is necessary to enable large-scale antennas to achieve significant gain.  \cite{Han2019ChannelEF,Cui2021ChannelEF,Cui2021NearFieldCE,Wei2022ChannelEF} designed codebooks for near field channel estimation utilizing spherical wavefronts. In \cite{Lu2022NearFieldMA}, the authors established the projected aperture non-uniform spherical wave (PNUSW) model and found that the multi-user interference (MUI) for XL-MIMO communications can be suppressed not only by angle separation but also by distance separation. Similarly to \cite{Lu2022NearFieldMA}, Wu et al. proposed the concept of location division multiple access (LDMA) for near field communication to provide a new possibility to enhance spectrum efficiency \cite{Wu2022MultipleAF}. In addition, 
\cite{Zhang2022BeamFF} encapsulated near field characteristics in steering vectors and studied the potential of beam focusing in near field multi-user MIMO communication scenarios with different antenna configurations.

 As the propagated signals in different angles and distances are separable, the number of paths greatly increases. This further can be considered as the increased degrees of freedom (DoFs). How to employ the increased DoFs from the transition from far field region to near field region arouses our interest. The idea that comes first and naturally is to exploit the spatial multiplexing gain by multi-data-stream transmission.
This requires the transmitter to equip with multiple radio frequency chains. However, the RF chains are costly and power-hungry. Recently, \cite{Wu2022} proposed distance-aware precoding for XL-MIMO. Different from previous hybrid precoding technology \cite{7397861} that activates a fixed number of RF chains, the precoding scheme proposed in \cite{Wu2022} dynamically activates the number of RF chains. When the communications distances is decreasing, more RF chains are selected to fully exploit the increasing spatial DoFs. They have shown that such a method can greatly enhance the spectral efficiency and their energy efficiency is comparable to a hybrid precoding scheme.
 

 The solution proposed in\cite{Wu2022} selects the best beamspace, whose dimension is equal to the number of activated RF chains, for data transmission.  Even though the RF chains can be dynamically chosen to be activated, the transceiver has to equip with many RF chains as candidates. 
In literature, there have been many studies on spatial modulation (SM) to reduce the number of RF chains and improve spectral efficiency. SM can be divided into two operating modes based on antenna switching \cite{Yang2017AdaptiveSF,Ishikawa2017GeneralizedSpatialModulationBased,He2017SpatialMF} and beamspace switching \cite{Wang2018TransmitSD,Ding2018BeamIM}. The former method cannot utilize the beamforming gain of all antennas, and the uniform activation scheme used in the latter is far from optimal since the good and bad channels are equiprobably activated. In our previous work \cite{Guo2019}, we proposed the basic concept of beamspace modulation (BM) and showed its superiority over  the best beamspace selection (BBS) or uniformly beamspace switching. We have also theoretically proved its capacity-achieving capabilities and superiority in a general MIMO channel with the number of spatial DoFs being greater than the number of RF chains in \cite{Guo2020}.   
How will BM perform in the near field region of XL-MIMO? Can it well explore the large and varying number of spatial DoFs in the near field region?

Motivated to answer these meaningful questions, this work has been done. In this paper, we propose the idea of using BM to exploit the increased DoFs in the near field region of XL-MIMO with a fixed small number of RF chains. Compared to the existing solution of only selecting the best beamspace for data transmission, beamspace hopping carries additional information. We theoretically analyze its superiority.
Both theoretical and simulation results show that the BM strategy greatly outperforms existing solutions over various communication distances.

\section{System Model}
\begin{figure*}[htpb]
       \centering
       \includegraphics[width=1\linewidth]{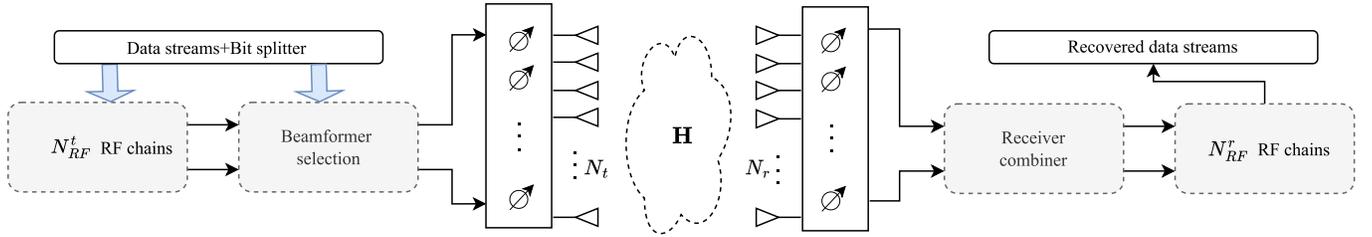}
       \caption{System model of a BM-based XL-MIMO communications systems}
       \label{Model}
\end{figure*}

In this paper, we consider the communication between transceivers with extremely large numbers of antennas. The numbers of transmitting and receiving antennas are denoted as $N_t$ and $N_r$, respectively.  The numbers of RF chains equipped at the transmitter and the receiver are represented by $N_{RF}^T$ and $N_{RF}^R$, respectively. Considering the power consumption of transmitting RF chains is much high, we assume only a few transmitting RF chains are equipped, leading to $N_{RF}^T$ being much smaller than $N_t$. Typically, $N_{RF}^T$ is set to be $1$ or $2$. As the receiving RF chains consume much less power. For simplicity, we assume $N_{RF}^R$ is equal to $N_r$.

\subsection{Channel Model and Spatial DoFs}
The channel model considered in this letter is a near-field two-ray channel model, which can be expressed as
\begin{equation}\label{eq1}
\mathbf{H}=\mathbf{H}_{\rm{LoS}}+\mathbf{H}_{\rm{NLoS}},
\end{equation}
where $\mathbf{H}_{\rm{LoS}}$ is the LoS component,  $\mathbf{H}_{\rm{NLoS}}$ is the first-order reflected path components of the scattering point.
The LoS component is given by
\begin{equation}\label{eq2}
\mathbf{H}_{\rm{LoS}}=\begin{bmatrix}
    \alpha_{11}e^{-j\frac{2\pi}{\lambda}r_{11}}&\cdots&\alpha_{1N_t}e^{-j\frac{2\pi}{\lambda}r_{1N_t}}\\
    \vdots&\ddots&\vdots\\
    \alpha_{N_r1}e^{-j\frac{2\pi}{\lambda}r_{N_r1}}&\cdots&\alpha_{N_rN_t}e^{-j\frac{2\pi}{\lambda}r_{N_rN_t}}
      \end{bmatrix},
\end{equation}
where $\alpha_{n_rn_t}$ and $r_{n_rn_t}$ represent the LoS channel gain and the distance between the $n_t$th transmitting antenna and the $n_r$-th receiving antenna, respectively. In (\ref{eq2}), the channel gain can be expressed as
\begin{equation}
    \alpha_{n_r,n_t}=\frac{\sqrt{G_tG_r}\lambda}{4\pi r_{n_r,n_t}},
\end{equation}
where $G_t$ and $G_r$ represent the antenna gains of the transmitting and receiving antennas, respectively. $\lambda$ stands for the wavelength of the carrier. The NLoS component is given by
\begin{equation}\label{eq4}
\mathbf{H}_{\rm{NLoS}}=\begin{bmatrix}
    \beta_{11}e^{-j\frac{2\pi}{\lambda}\tilde{r}_{11}}&\cdots&\beta_{1N_t}e^{-j\frac{2\pi}{\lambda}\tilde{r}_{1N_t}}\\
    \vdots&\ddots&\vdots\\
    \beta_{N_r1}e^{-j\frac{2\pi}{\lambda}\tilde{r}_{N_r1}}&\cdots&\alpha_{N_rN_t}e^{-j\frac{2\pi}{\lambda}\tilde{r}_{N_rN_t}}
      \end{bmatrix},
\end{equation}
where $\beta_{n_rn_t}$ and $\tilde{r}_{n_rn_t}$ represent the NLoS channel gain and the distance from the $n_t$-th transmit antenna to the point of scattering  to the $n_r$-th receive antenna, respectively.
In (\ref{eq4}), the channel gain can be expressed as
\begin{equation}
    \beta_{n_r,n_t}=\frac{\Gamma\sqrt{G_tG_r}\lambda}{4\pi \tilde{r}_{n_r,n_t}},
\end{equation}
where $\Gamma$ is the reflection coefficient.

In order to demonstrate the variation of the near field two-ray channel DoFs as the distance
between the transmitter and receiver varies, we consider a uniform linear array of 256 antennas
at half-wavelength intervals for both transmitter and receiver. We assume that large antenna
arrays at the transmitter and receiver sides are parallel. Moreover, we assume that there is a scattering point in the middle of the transceiver antenna that does not affect the LOS path and generates an NLOS path. 
Based on the channel given by (1), the number of non-zero singular values of the channel matrix $\mathbf{H}$ are counted as the spatial DoFs, i.e., $N_{\rm{DoF}}=\operatorname{Rank}\left\{\mathbf{H}\right\}$.

The spatial DoFs are depicted in Fig. \ref{DoFs}.  We illustrate the variation of DoFs at three frequencies of $5$ gigahertz (GHz), $30$ GHz, and $100$ GHz in Fig. 2. It is clear that as the distance between the transmitter and receiver decreases, the spatial DoFs increase rapidly. 

\setlength{\abovecaptionskip}{-0.05cm}
\begin{figure}[!t]
       \centering
       \includegraphics[width=1\linewidth]{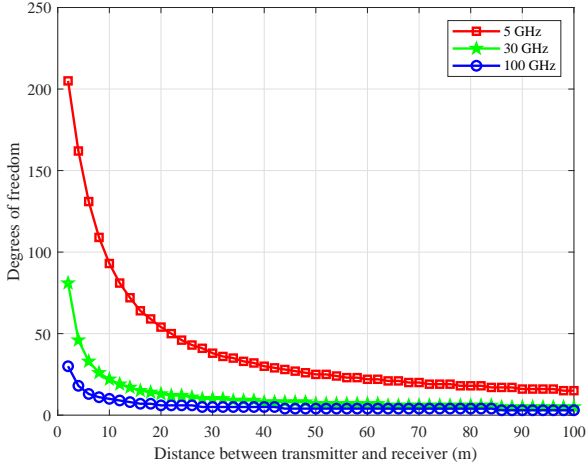}
       \caption{DoFs variance as the transceiver distance varies.}
       \label{DoFs}
       \vspace{-0.2cm}
\end{figure}


\subsection{Existing Beamforming Strategy}
Even though spatial DoFs are sufficiently large, especially in the near field region, the number of spatial DoFs that can be used is limited by the number of transmit RF chains, i.e., $N_{RF}^T$. The conventional solution employs the beamformer $\mathbf{F}_{\rm{opt}}\in \mathbb{C}^{N_t\times N_{RF}^T}$ that corresponds to the strongest beamspace for data transmission, whose transmission model can be expressed by
\begin{equation}\label{eq10}
\mathbf{y}=\mathbf{H}\mathbf{F}_{\rm{opt}}\mathbf{s}+\mathbf{n},
\end{equation}
where $\mathbf{y}$ stands for the received signal; $\mathbf{s}$ represents the data symbol vector $\mathbf{s}\in\mathbb{C}^{N_{RF}^T\times 1}$ with $\mathbb{E}\{\mathbf{s}\mathbf{s}^H\}=\frac{1}{{N}_{RF}^T}\mathbf{I}_{N_{RF}^T}$ and $\mathbf{n}$ represents the complex Gaussian noise with zero mean and covariance $\sigma_n^2 \mathbf{I}_{N_r}$. 
The BBS solution of $\mathbf{F}_{\rm{opt}}$ consists of the right singular vectors corresponding to the largest $N_{RF}^T$ singular values. This indicates the best beamspace of dimension $N_{RF}^T$ is selected for data transmission.  The other $(N_{\rm{DoF}}-N_{RF}^T)$ DoFs are left unexplored.

\section{Beamspace Modulation and Its Capacity}
To employ all spatial DoFs, we introduce beasmspace modulation and analyze its capacity in the section. As illustrated in Fig. 1, the information bits are first split into two parts, one part for traditional symbol modulation, which maps to the data symbol vector $\mathbf{s}$, and the other part for beamformer selection. The beamformer candidates have to be independent of each other, which means a beamformer matrix should not be expressed by the linear combination of other beamformer matrices. There are a total $N_{\rm{DoF}}$ independent singular vectors, $N_{RF}^T$ of which are chosen to form a beamformer $\mathbf{F}_i$, leading to a total of $K=\left(N_{\rm{DoF}}\atop N_{RF}^T\right)$ beamformer candidates. Let
 $\mathcal{F}=\{\mathbf{F}_1,\mathbf{F}_2,\cdots,\mathbf{F}_{K}\}$ represent the beamformer set and $p_i$ denote the activating probability of $\mathbf{F}_i$. The transmission model can be expressed by 
 \begin{equation}\label{eq13}
 \mathbf{y}=\mathbf{H}\mathbf{F}_i\mathbf{s}+\mathbf{n},
 \end{equation}
where information bits are not only carried by $\mathbf{s}$ but also by $\mathbf{F}_i$. In this way, the spectral efficiency characterized by the mutual information $\mathcal{I}(\mathbf{F}_i,\mathbf{s};\mathbf{y})$ can be expressed by 
\begin{equation}\label{eq14}
\begin{split}\mathcal{I}(\mathbf{F}_i,\mathbf{s};\mathbf{y})=\mathcal{I}(\mathbf{s};\mathbf{y})+\mathcal{I}(\mathbf{F}_i;\mathbf{y}|\mathbf{s})\geq\mathcal{I}(\mathbf{s};\mathbf{y}),
\end{split}
\end{equation}
which implies that we can improve the mutual information limit by using $\mathbf{F}_i$ carrying information. In this letter, the set of beams is the set of $K$ candidate beams selected from the right singular vector of the channel $\mathbf{H}$, which means $\mathcal{F}$ is deterministic when the channel is known. Then SE $\mathcal{R}(\mathbf{p})$ of BM can be derived as 
\begin{equation}\label{eq15}
\begin{split}
\mathcal{R}(\mathbf{p})&=\mathcal{I}(\mathbf{F}_i,\mathbf{s};\mathbf{y})\\
&=\mathcal{H}(\mathbf{y})-\mathcal{H}(\mathbf{y}|\mathbf{F}_i,\mathbf{s})\\
&=\mathbb{E}\left[-\log_2 f(\mathbf{y})\right]-\mathcal{H}(\mathbf{n})\\
&=\int_{\mathbb{C}^{N_r}}f(\mathbf{y})\log_2 f(\mathbf{y}) d\mathbf{y}-N_r\log_2 (\pi e),
\end{split}
\end{equation}
where $\mathbf{p}=[p_1,p_2,\cdots,p_K]^T$; $\mathcal{H}(\cdot)$ represents the entropy function and $\mathbb{E}[\cdot]$ denotes the expectation. It is shown that the spectral efficiency is determined by the distribution of the received signal vector, i.e., $f(\mathbf{y})$. 
According to the Theorem 1 in \cite{Ibrahim2016OnTA}, the zero-mean Gaussian assumption of the symbol 
vector $\mathbf{s}$ fed into the precoder leads to that  the  received vector $\mathbf{y}$ is distributed as a complex Gaussian mixture model (GMM) with probability density function (PDF) being expressed as 
\begin{equation}\label{eq16}
f(\mathbf{y})=\sum_{i=1}^{K}p_i f_i(\mathbf{y}),
\end{equation}
where $p_i$ is the probability of $\mathbf{F}_i$ which satisfies the constraint  $\sum_{i=1}^K p_i=1$.
In (\ref{eq16}),
$f_i(\mathbf{y})$ is the probability of PDF of complex Gaussian distribution, which can be expressed as
\begin{equation}\label{eq17}
p(\mathbf{y}|\mathbf{F=F_i})=f_i(\mathbf{y})=\frac{1}{\pi^{N_r}\det(\Sigma_i)}\exp(-\mathbf{y}^H\Sigma_i^{-1}\mathbf{y}),
\end{equation}
and $\Sigma_i=\mathbf{I}_{N_r}+\frac{1}{N_{RF}^T\sigma_n^2}\mathbf{H}_i\mathbf{F}_i\mathbf{F}_i^H\mathbf{H}_i^H$.
\subsection{Problem Formulation}
Then we formulate the channel capacity  as 
\begin{equation}\label{eq18}
    \mathcal{C}_{\rm{BM}}=\max_{\sum_{i=1}^{K}p_i=1}{\mathcal{R}}{(\mathbf{p})}.
\end{equation}
It is difficult to solve the problem since its expression involves the integration of complex functions. Fortunately, the upper bound on the entropy of the real GMM random vector has been proved in \cite{Huber2008OnEA}, and we can use a similar derivation to obtain an upper bound on the entropy of the complex GMM random vector. Specifically, the upper bound on the differential entropy of the received complex GMM random vector y, which is the first term in (\ref{eq15}), is
\begin{equation}\label{eq19}
    \mathcal{H}(\mathbf{y})\leq\sum_{i=1}^{K} p_{i}\left(-\log p_{i}+\log \left[(\pi e)^{N_{r}} \operatorname{det}\left(\Sigma_{i}\right)\right]\right).
\end{equation}
Combine (\ref{eq15}) and (\ref{eq19}), an upper bound for the SE (denoted by $\widetilde{\mathcal{R}}(\mathbf{p})$) can be derived as
\begin{equation}\label{eq20}
\begin{split}
    \mathcal{R}(\mathbf{p})\le\widetilde{\mathcal{R}}(\mathbf{p})
    &=-\sum_{i=1}^{K} p_{i} \log p_{i}+\sum_{i=1}^{K}\left(p_{i} \log \operatorname{det}\left(\Sigma_{i}\right)\right)\\&+N_{r} \log (\pi e)\sum_{i=1}^{K}p_{i}  -N_{r} \log (\pi e).\\
    \end{split}
\end{equation}
Following the fact that $\sum_{i=1}^K p_i=1$, we finally obtain 
\begin{equation}\label{eq21}
    \widetilde{\mathcal{R}}(\mathbf{p})=\sum_{i=1}^{K}p_i[\log_2 \det (\Sigma_i)-\log_2 p_i].
\end{equation}
 In \cite{Guo2020}, it has been proved that the upper bound in (\ref{eq21}) is tight in high signal-to-noise ratio (SNR) regimes. Therefore, we resort to optimizing the asymptotic spectral efficiency instead and formulate the optimization problem as
 \begin{equation}\label{eq22}
 \begin{split} \mathcal{C}_{\rm{BM}}^A&=\max_{\sum_{i=1}^{K}p_i=1}\widetilde{\mathcal{R}}{(\mathbf{p})}\\
&=\max_{\sum_{i=1}^{K}p_i=1} \sum_{i=1}^{K}p_i\left[\log_2 \det (\Sigma_i)-\log_2 p_i\right], 
\end{split}
 \end{equation}
 where $\mathcal{C}_{\rm{BM}}^A$ is the asymptotic capacity.
\subsection{Beamspace Activation Optimization} 
Problem (\ref{eq22}) is a typical constrained optimization problem that can be solved using the Lagrange multiplier method. We formulate the Lagrange multiplier for (\ref{eq22}) as
\begin{equation}\label{eq23}
     \mathcal{L}(\mathbf{p},\lambda)=\sum_{i=1}^{K}p_i\left[\log_2 \det (\Sigma_i)-\log_2 p_i\right]-\lambda\left(\sum_{i=1}^K p_i-1\right).
\end{equation}
Taking the derivation of $\mathcal{L}(\mathbf{p},\lambda)$ with respect to $p_i$ yields
 \begin{equation}\label{eq24}
 \log_2 \det(\Sigma_i)-\log_2 p_i -\frac{1}{\ln 2}-\lambda =0, ~i=1,2,\cdots, K.
 \end{equation}
 By solving (\ref{eq24}), we can obtain
\begin{equation}\label{eq25}
p_i=\frac{2^{-\lambda}\det(\Sigma_i)}{e},~i=1,2,\cdots,K.
\end{equation}
Owing to $\sum_{i=1}^K p_i=1$, we can derive the capacity-achieving beamformer activation probability as\footnote{From (\ref{eq26}), it is shown that the beamformer activation probabilities are non-equal. 
To transform independent and Bernoulli($1/2$) distributed input bits into a sequence of output symbols with a desired distribution,  distribution matching techniques \cite{Schulte2016} can be adopted.}
\begin{equation}\label{eq26}
p_i^*=\frac{p_i}{\sum_{i=1}^Kp_i}=\frac{\det(\Sigma_i)}{\sum_{i=1}^{K}\det(\Sigma_i)},~i=1,2,\cdots,K,
\end{equation}
Substituting (\ref{eq26}) into (\ref{eq21}), we can obtain the the corresponding asymptotic  capacity as\begin{equation}\label{eq27}
\begin{split}
\mathcal{C}_{\rm{BM}}^A&=\sum_{i=1}^{K}p_i^*[\log_2 \det (\Sigma_i)-\log_2 p_i^*]\\&=\sum_{i=1}^{K}p_i^*\log_2{\left[\sum_{i=1}^{K}\det(\Sigma_i)\right]}=\log_2{\left[\sum_{i=1}^{K}\det(\Sigma_i)\right]}.
\end{split}
\end{equation}

By comparing BM with the conventional transmission solution that selects the best beamspace for conveying information, which actually achieves the capacity as
\begin{equation}\label{eq28}
\mathcal{C}_{\rm{BBS}}=\log_2{\left[\max_{i=1,2,\cdots, K}\det(\Sigma_i)\right]},
\end{equation}
we can easily conclude that the BM solution theoretically outperforms
the BBS solution in spectral efficiency. Moreover, we can make the following corollary.
\begin{corollary}
The closer the transceiver is, the more significant the gain in spectral efficiency of BM compared to BBS. As the communication distance increases, BM will always outperform BBS in the high SNR regime as long as the number of  channel DoFs is higher than the number of RF chains.
\end{corollary}
\begin{proof}
According to the analysis in Section II-A, the increase in the distance leads to a decrease in DoFs. When the distance is far enough, the size of beamformer candidates $K\to 1$ and $\mathbf{F}_i=\mathbf{F}_{opt}$, then $\mathcal{C}_{\rm{BM}}^A\approx\mathcal{C}_{\rm{BBS}}$ in the high SNR regime.\\
\end{proof}

\subsection{Complexity Analysis}
The computational complexity of BM comes from the SVD and the beamspace activation probability optimization. The computational complexity of the SVD is $O\left(N_{t} N_{r} \min \left(N_{t}, N_{r}\right)\right)$. The computational complexity of beamspace activation probability optimization mainly comes from the computation of $\left\{\operatorname{det}\left(\Sigma_{i}\right)\right\}$ including matrix multiplication and determinant calculation, which is about $O\left(K N_{r}^{3}+K N_{t} N_{r} N_{RF}^T+K N_{r}^{2} N_{RF}^T\right)$.  Therefore, the aggregated computational complexity can be expressed as $O\left(N_{t} N_{r} \min \left(N_{t}, N_{r}\right)+K N_{r}^{3}+K N_{t} N_{r} N_{RF}^T+K N_{r}^{2} N_{RF}^T\right)$.

\section{Simulations and Discussions}
In this section,   simulations are conducted to validate the capacity improvement of BM. In the simulations, the carrier frequency is set to be $30$ GHz, and half-wavelength spaced parallel linear antenna arrays with $256$ antennas are equipped at the transceivers. The number of transmit RF chains are set to be $1$.
\begin{figure}[t]
       \centering
       \includegraphics[width=1\linewidth]{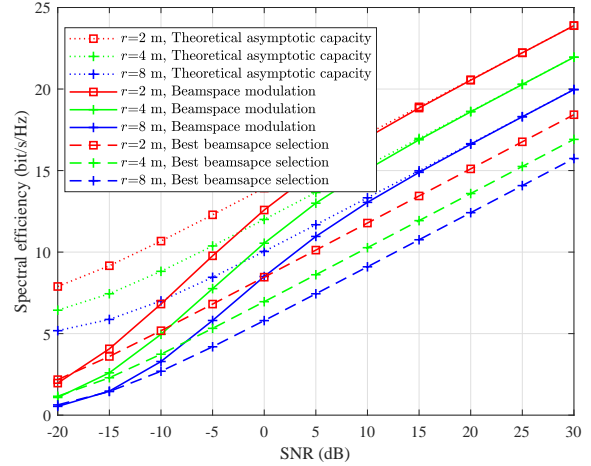}
       \caption{Spectral efficiency comparison between BM and BBS.} \label{result1}
       \vspace{-0.2cm}
\end{figure}

\begin{figure}[t] 
       \centering
       \includegraphics[width=1\linewidth]{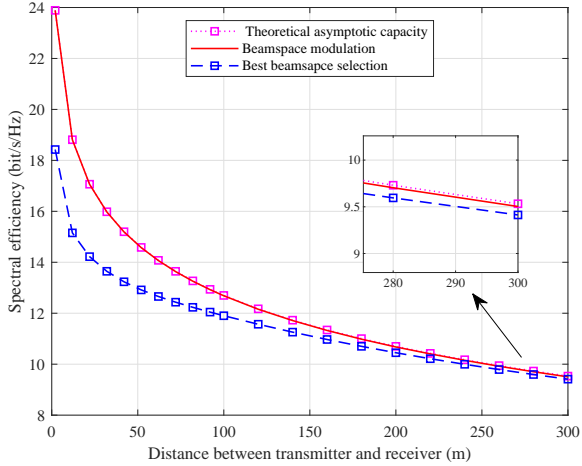}
       \caption{Spectral efficiency versus transceiver distances at an SNR of $30$ dB.}
       \label{result2}
       \vspace{-0.2cm}
\end{figure}
The spectral efficiency of BM is firstly simulated at the SNR from  $-20$ dB to $30$ dB in the near field region. Distances between the transmitter and receiver are set as $2$, $4$, and $8$ meters, respectively. All simulation results are illustrated in Fig. \ref{result1}. The real achieved spectral efficiency in (\ref{eq15}) is evaluated by utilizing Monte Carlo simulations to compute the expectation.  To verify the correctness of the analysis of the theoretical asymptotic capacity, we also the numerical results of (\ref{eq27}). Simulation results show that the asymptotic capacity matches the simulation results well in the high SNR regime above $20$ dB. For comparison, we also include the spectral efficiency of the conventional transmission solution, i.e., the BBS solution.  It is shown in  Fig. 3 that the spectral efficiency of BM outperforms that of BBS by around  $5.44$, $5.02$, and $4.15$ bits/s/Hz at the transceiver distance of $2$ meters, $4$ meters, and $8$ meters in high SNR regime, respectively. It is noteworthy that these considerable performance gains are achieved  without increasing any power or changing the hardware. Moreover, it is shown the performance gain goes higher as the receiver moves closer to the transmitter.

To show the detailed variation of the performance gain brought by BM along with the decrease in distance between the transmitter and receiver, we simulate the spectral efficiency of BM, its theoretical asymptotic capacity, and the spectral efficiency of BBS from  $2$ meters to $300$ meters at a high SNR of $30$ dB. Simulation results are illustrated in Fig. \ref{result2} It is shown that the performance gain is high in the near field region and gradually shrinks as the receiver moves further. Even though the performance gain shrinks, it is still significant (around $0.12$ bit/s/Hz at a distance of $300$ meters) in the far field region. Observing the results in Fig. \ref{result2}, we find that as the user moves from far field to hybrid field and to near field, it will always benefit from the BM. This is appealing for future communication networks utilizing mmWave and Terahertz bands.

\section{Conclusion}
In this paper, we proposed to utilize BM for capacity improvement in the near field region. The beamformers are activated with non-uniformly distributed probabilities. The closed-form activation probabilities were derived.
Given the limited number of RF chains, it was shown
that BM can achieve considerable performance improvement not only in the near field region but also in the far field region.
It is appealing that the capacity achievement is not at the expense of modifying the hardware infrastructure and may become a promising technology in 6G.

\bibliographystyle{IEEEtran} 
\bibliography{IEEEabrv,bib}

\begin{thebibliography}{10}
\providecommand{\url}[1]{#1}
\csname url@samestyle\endcsname
\providecommand{\newblock}{\relax}
\providecommand{\bibinfo}[2]{#2}
\providecommand{\BIBentrySTDinterwordspacing}{\spaceskip=0pt\relax}
\providecommand{\BIBentryALTinterwordstretchfactor}{4}
\providecommand{\BIBentryALTinterwordspacing}{\spaceskip=\fontdimen2\font plus
\BIBentryALTinterwordstretchfactor\fontdimen3\font minus
  \fontdimen4\font\relax}
\providecommand{\BIBforeignlanguage}[2]{{%
\expandafter\ifx\csname l@#1\endcsname\relax
\typeout{** WARNING: IEEEtran.bst: No hyphenation pattern has been}%
\typeout{** loaded for the language `#1'. Using the pattern for}%
\typeout{** the default language instead.}%
\else
\language=\csname l@#1\endcsname
\fi
#2}}
\providecommand{\BIBdecl}{\relax}
\BIBdecl

\bibitem{dang2020should}
S.~Dang, O.~Amin, B.~Shihada, and M.-S. Alouini, ``What should 6{G} be?''
  \emph{Nat. Electron.}, vol.~3, no.~1, pp. 20--29, Jan. 2020.

\bibitem{Selvan2017FraunhoferAF}
K.~T. Selvan and R.~Janaswamy, ``Fraunhofer and fresnel distances : Unified
  derivation for aperture antennas.'' \emph{{IEEE} Antennas Propag. Mag.},
  vol.~59, pp. 12--15, Aug. 2017.

\bibitem{Cui2022NearFieldCF}
M.~Cui, Z.~Wu, Y.~Lu, X.~Wei, and L.~Dai, ``Near-field {MIMO} communications
  for 6g: Fundamentals, challenges, potentials, and future directions,''
  \emph{IEEE Commun. Mag.}, vol.~61, no.~1, pp. 40--46, Jan. 2023.

\bibitem{Cui2021ChannelEF}
M.~Cui and L.~Dai, ``Channel estimation for extremely large-scale {MIMO}:
  Far-field or near-field?'' \emph{{IEEE} Trans. Commun.}, vol.~70, pp.
  2663--2677, Aug. 2021.

\bibitem{Han2019ChannelEF}
Y.~Han, S.~Jin, C.-K. Wen, and X.~Ma, ``Channel estimation for extremely
  large-scale massive {MIMO} systems,'' \emph{{IEEE} Wireless Commun. Lett.},
  vol.~9, pp. 633--637, Oct. 2019.

\bibitem{Cui2021NearFieldCE}
M.~Cui and L.~Dai, ``Near-field channel estimation for extremely large-scale
  {MIMO} with hybrid precoding,'' \emph{2021 IEEE Global Communications
  Conference (GLOBECOM)}, pp. 1--6, Dec. 2021.

\bibitem{Wei2022ChannelEF}
X.~Wei and L.~Dai, ``Channel estimation for extremely large-scale {Massive}
  {MIMO}: Far-field, near-field, or hybrid-field?'' \emph{{IEEE} Commun.
  Lett.}, vol.~26, pp. 177--181, Jan. 2022.

\bibitem{Lu2022NearFieldMA}
H.~Lu and Y.~Zeng, ``Near-field modeling and performance analysis for
  multi-user extremely large-scale {MIMO} communication,'' \emph{{IEEE} Commun.
  Lett.}, vol.~26, pp. 277--281, Feb. 2022.

\bibitem{Wu2022MultipleAF}
Z.~Wu, M.~Cui, and L.~Dai, ``Multiple access for near-field communications:
  {SDMA} or {LDMA}?'' \emph{arXiv:2208.06349}, Aug. 2022.

\bibitem{Zhang2022BeamFF}
H.~Zhang, N.~Shlezinger, F.~Guidi, D.~Dardari, M.~F. Imani, and Y.~C. Eldar,
  ``Beam focusing for near-field multiuser {MIMO} communications,''
  \emph{{IEEE} Trans. Wireless Commun.}, vol.~21, pp. 7476--7490, May 2022.

\bibitem{Wu2022}
Z.~Wu, M.~Cui, Z.~Zhang, and L.~Dai, ``\BIBforeignlanguage{en}{Distance-aware
  precoding for near-field capacity improvement in {XL}-{MIMO}},'' in
  \emph{\BIBforeignlanguage{en}{2022 {IEEE} 95th {Vehicular} {Technology}
  {Conference}: ({VTC2022}-{Spring})}}.\hskip 1em plus 0.5em minus 0.4em\relax
  Helsinki, Finland: IEEE, June. 2022, pp. 1--5.

\bibitem{7397861}
X.~Yu, J.-C. Shen, J.~Zhang, and K.~B. Letaief, ``Alternating minimization
  algorithms for hybrid precoding in millimeter wave {MIMO} systems,''
  \emph{{IEEE} J. Sel. Topics Signal Process.}, vol.~10, no.~3, pp. 485--500,
  Apr. 2016.

\bibitem{Yang2017AdaptiveSF}
P.~Yang, Y.~Xiao, Y.~L. Guan, Z.~Liu, S.~Li, and W.~Xiang, ``Adaptive {SM-MIMO}
  for mmwave communications with reduced {RF} chains,'' \emph{{IEEE} J. Sel.
  Areas Commun.}, vol.~35, pp. 1472--1485, Jul. 2017.

\bibitem{Ishikawa2017GeneralizedSpatialModulationBased}
N.~Ishikawa, R.~Rajashekar, S.~Sugiura, and L.~Hanzo,
  ``Generalized-spatial-modulation-based reduced-{RF}-chain millimeter-wave
  communications,'' \emph{{IEEE} Trans. Veh. Technol.}, vol.~66, pp. 879--883,
  Jan. 2017.

\bibitem{He2017SpatialMF}
L.~He, J.~Wang, and J.~Song, ``Spatial modulation for more spatial
  multiplexing: {RF}-chain-limited generalized spatial modulation aided mm-wave
  {MIMO} with hybrid precoding,'' \emph{{IEEE} Trans. Commun.}, vol.~66, pp.
  986--998, Mar. 2017.

\bibitem{Wang2018TransmitSD}
W.~Wang and W.~Zhang, ``Transmit signal designs for spatial modulation with
  analog phase shifters,'' \emph{{IEEE} Trans. Wireless Commun.}, vol.~17, pp.
  3059--3070, May 2018.

\bibitem{Ding2018BeamIM}
Y.~Ding, V.~F. Fusco, A.~P. Shitvov, Y.~Xiao, and H.~Li, ``Beam index
  modulation wireless communication with analog beamforming,'' \emph{{IEEE}
  Trans. Veh. Technol.}, vol.~67, pp. 6340--6354, Jul. 2018.

\bibitem{Guo2019}
S.~Guo, H.~Zhang, P.~Zhang, P.~Zhao, L.~Wang, and M.-S. Alouini, ``Generalized
  beamspace modulation using multiplexing: A breakthrough in {mmWave} {MIMO},''
  \emph{{IEEE} J. Sel. Areas Commun.}, vol.~37, no.~9, pp. 2014--2028, Sep.
  2019.

\bibitem{Guo2020}
S.~Guo, H.~Zhang, and M.-S. Alouini, ``Asymptotic capacity for {MIMO}
  communications with insufficient radio frequency chains,'' \emph{{IEEE}
  Trans. Commun.}, vol.~68, no.~7, pp. 4190--4201, Jul. 2020.

\bibitem{Ibrahim2016OnTA}
A.~A.~I. Ibrahim, T.~Kim, and D.~J. Love, ``On the achievable rate of
  generalized spatial modulation using multiplexing under a {G}aussian mixture
  model,'' \emph{{IEEE} Trans. Commun.}, vol.~64, pp. 1588--1599, Apr. 2016.

\bibitem{Huber2008OnEA}
M.~F. Huber, T.~Bailey, H.~F. Durrant-Whyte, and U.~D. Hanebeck, ``On entropy
  approximation for gaussian mixture random vectors,'' \emph{2008 IEEE
  International Conference on Multisensor Fusion and Integration for
  Intelligent Systems}, pp. 181--188, Oct. 2008.

\bibitem{Schulte2016}
P.~Schulte and G.~Böcherer, ``Constant composition distribution matching,''
  \emph{IEEE Trans. Inf. Theory}, vol.~62, no.~1, pp. 430--434, Jan. 2016.

\end{thebibliography}
\end{document}